\documentclass[11pt,amssymb,amsfont,a4paper]{article}
\usepackage{enumerate}
\usepackage{amssymb}
\usepackage{latexsym}
\usepackage{amsmath}
\usepackage{hyperref,cite}
\usepackage{amsthm}
\usepackage{empheq}
\usepackage{bm}
\usepackage{authblk} 
\usepackage{comment}
\numberwithin{equation}{section}
\theoremstyle{definition}
\newtheorem{theorem}{Theorem}[section]
\newtheorem{corollary}[theorem]{Corollary}

\newtheorem{definition}[theorem]{Definition}

\newtheorem{claim}[theorem]{Claim}

\newtheorem{remark}[theorem]{Remark}
\newtheorem{lemma}[theorem]{Lemma}

\newcommand{\numberset}{\mathbb}
\newcommand{\N}{\numberset{N}}

\newcommand{\cod}{\mathcal{C}}

\newcommand{\F}{\numberset{F}}

\newcommand{\Fq}{\F_q}
\newcommand{\Fqm}{\F_{q^m}}
\newcommand{\Fqn}{\F_{q^n}}
\newcommand{\Fqmn}{\F_{q^m}^n}

\newcommand{\Fh}{\F_h}

\newcommand{\X}{\boldsymbol{X}}
\newcommand{\Y}{\boldsymbol{Y}}
\newcommand{\balpha}{\boldsymbol{a}}
\newcommand{\bbeta}{\boldsymbol{\beta}}
\newcommand{\Xs}{\boldsymbol{x}}
\newcommand{\Ys}{\boldsymbol{y}}
\newcommand{\wsr}{\mathrm{wt}_{\mathsf{SR}}}
\newcommand{\dsr}{d_{\mathsf{SR}}}
\newcommand{\Bsr}{B_{\mathsf{SR}}}
\newcommand{\Ssr}{S_{\mathsf{SR}}}

\newcommand{\rank}{\text{rank}}
\newcommand{\poly}{\text{poly}}
\newcommand{\spa}{\text{span}}

\newcommand{\os}{1,2,\cdots,}

\newcommand{\eps}{\varepsilon}
\title{Explicit List-Decodable Linearized Reed--Solomon and Folded Linearized Reed--Solomon Subcodes
\thanks{
K. Shang is with School of Computer Science, Shanghai Jiao Tong University. (Email: \href{billy63878@sjtu.edu.cn}{billy63878@sjtu.edu.cn})  C. Yuan is with School of Computer Science, Shanghai Jiao Tong University. (Email: \href{chen_yuan@sjtu.edu.cn}{chen\_yuan@sjtu.edu.cn}) R. Zhu is with School of Computer Science, Shanghai Jiao Tong University. (Email: \href{sjtuzrq7777@sjtu.edu.cn}{sjtuzrq7777@sjtu.edu.cn})
}}
\author{Kuo Shang, Chen Yuan, Ruiqi Zhu}
\date{\today}

\begin{document}

\maketitle

\begin{abstract}
The sum-rank metric is the mixture of the Hamming and rank metrics. The sum-rank metric found its application in network coding, locally repairable codes, space-time coding, and quantum-resistant cryptography. Linearized Reed--Solomon (LRS) codes are the sum-rank analogue of Reed--Solomon codes and strictly generalize both Reed–Solomon and Gabidulin codes.

In this work, we construct an explicit family of $\mathbb{F}_h$-linear sum-rank metric codes over arbitrary fields $\F_h$. Our construction enables efficient list decoding up to a fraction $\rho$ of errors in the sum-rank metric with rate $1-\rho-\varepsilon$, for any desired $\rho \in (0,1)$ and $\varepsilon>0$. Our codes are subcodes of LRS codes, obtained by restricting message polynomials to an $\F_h$-subspace derived from subspace designs, and the decoding list size is bounded by $h^{\mathrm{poly}(1/\eps)}$. 

Beyond the standard LRS setting, we further extend our linear-algebraic decoding framework to folded Linearized Reed--Solomon (FLRS) codes. We show that folded evaluations satisfy appropriate interpolation conditions and that the corresponding solution space forms a low-dimensional, structured affine subspace. This structure enables effective control of the list size and yields the first explicit positive-rate FLRS subcodes that are efficiently list decodable beyond the unique-decoding radius. To the best of our knowledge, this also constitutes the first explicit construction of positive-rate sum-rank metric codes that admit efficient list decoding beyond the unique decoding radius, thereby providing a new general framework for constructing efficiently decodable codes under the sum-rank metric.
\end{abstract}

\section{Introduction}
The sum-rank metric, which unifies both the Hamming and rank metrics, has found widespread application in network coding and related communication systems. Sum-rank metric codes, introduced by N\'obrega and Uch\^oa-Filho~\cite{nobrega2010multishot}, were first used to reduce the size of the network alphabet. A sum-rank metric code consists of $\ell$ matrices, where for each index $i$ ($1 \le i \le \ell$), the matrix has size $m_i\times n_i$ over a finite field $\F_q$. The distance $d_{\mathsf{SR}}$ between two codewords is the sum, over all coordinate matrices, of the ranks of their differences. Sum-rank metric codes can be viewed as a mixture of Hamming-metric and rank-metric codes. In particular, they degenerate to Hamming-metric codes when $n_i = m_i = 1$ for all $i$, and to rank-metric codes when $\ell = 1$. The most extensively studied regime corresponds to the uniform parameter setting $n_1 = n_2 = \cdots = n_\ell = n$ and $m_1 = m_2 = \cdots = m_\ell = m$. In this case, each block can be naturally identified with an element of $\F_{q^m}^n$ (see Section~\ref{sec:pre}). This is the parameter regime on which we focus throughout this work.

In this work, we focus on subcodes of linearized Reed--Solomon (LRS) codes, a class of sum-rank metric codes introduced by Mart\'inez-Pe\~nas~\cite{martinez2018skew}. LRS codes extend classical Reed--Solomon codes to the sum-rank metric through the evaluation of skew polynomials over extension fields, thereby unifying and generalizing the algebraic frameworks underlying both Reed--Solomon and Gabidulin codes. LRS codes have found applications in multi-shot network coding~\cite{martinez2019reliable, nobrega2010multishot}, locally repairable codes~\cite{martinez2019universal}, space-time coding~\cite{lu2005unified}, and post-quantum cryptography~\cite{puchinger2022generic}.

An LRS code is constructed by encoding skew polynomials of degree less than $k$ via their generalized operator evaluations $f(\beta)_a$ (see Definition~\ref{def:generalized_operator_evaluation}), resulting in a sum-rank metric code of rate $R = k/n$ and minimum sum-rank distance $d_{\mathsf{SR}} = n - k + 1$. Since LRS codes attain the Singleton bound under the sum-rank metric, they are referred to as maximum sum-rank distance (MSRD) codes.

Beyond the classical LRS setting, folded versions of Reed--Solomon codes have played a central role in achieving list-decoding capacity under the Hamming metric~\cite{parvaresh2005correcting, guruswami2008explicit, ashvinkumar2025algorithmic, chen2025explicit}. The folding operation groups several evaluations together to create stronger interpolation constraints, significantly improving list-decoding performance. Motivated by these developments, folded linearized Reed--Solomon (FLRS) codes were introduced as the natural sum-rank metric analogs of folded Reed–-Solomon codes~\cite{hormann2021efficient, hormann2024interpolation}. 
FLRS codes inherit the skew-polynomial evaluation structure of LRS codes while benefiting from the algebraic gains of folding, making them promising candidates for list decoding beyond the unique-decoding radius under the sum-rank metric.

However, while folded Reed--Solomon codes under the Hamming metric and folded Gabidulin codes under the rank metric have been extensively studied~\cite{mahdavifar2012list}, the list decoding problem for FLRS codes in the general sum-rank metric remains largely unexplored. Existing works on FLRS codes primarily focus on unique decoding or restricted structural variants, and no explicit constructions are currently known that admit efficient list decoding beyond the unique-decoding radius.

In parallel, list decoding in the sum-rank metric is far less understood than in the Hamming or rank metrics. Although significant progress has been made on list decoding for rank-metric codes~\cite{trombetti2020list, xing2017new, guo2024random, wachter2013bounds, ding2014list}, very little is known about explicit constructions that can be efficiently list-decoded beyond half the minimum distance in the sum-rank metric. This gap motivates the study of explicit LRS and FLRS constructions that admit efficient list decoding approaching the list-decoding capacity.

     \subsection{Related Works}
     Building upon the progress on list decoding in the rank-metric, recent research has increasingly focused on the sum-rank metric, which generalizes both rank and Hamming metrics and plays a central role in modern coding applications. 
     Significant efforts in recent years has explored the fundamental coding-theoretical properties of sum-rank metric codes~\cite{byrne2021fundamental,ott2021bounds,martinez2019theory,camps2022optimal,ott2022covering}. 
     
     There are fruitful results on the constructions of LRS codes and efficient decoding algorithms for them~\cite{boucher2020algorithm,bartz2021fast,hormann2021efficient,hormann2022error}. In~\cite{bartz2021fast}, the authors employ skew polynomials to accelerate the decoding algorithm for linearized Reed--Solomon codes in the sum-rank metric. 
     Recently, ~\cite{liu2025list} established a list-decoding capacity theorem for sum-rank metric codes. In particular, a random general sum-rank code of rate below capacity is $(\rho,O(\frac{1}{\eps}))$-list-decodable with high probability, where $\rho\in (0,1)$ and $\eps>0$. 
     
     In contrast, the list-decodability of explicit sum-rank code families has not been thoroughly understood. Bounds on the list sizes of LRS codes were investigated in~\cite{puchinger2021bounds} where it was also shown that certain families of LRS codes cannot be list-decoded beyond the unique decoding radius. Chen~\cite{chen2021list} derived upper bounds on the achievable rate of list-decodable sum-rank metric codes in the uniform setting $m_i = m$ and $\eta_i = \eta$ for all $i = 1, \ldots, \ell$. Several list-decoding algorithms have been proposed for interleaved linearized Reed--Solomon codes~\cite{bartz2021decoding,hormann2024syndrome}, folded linearized Reed--Solomon codes~\cite{hormann2024interpolation} and lifted interleaved linearized Reed--Solomon~\cite{bartz2024fast}. However, for constant-rate codes, the output list size of these algorithms typically grows exponentially with the code length~$n$.

     Beyond these developments, techniques from list decoding in the Hamming metric have also influenced progress in the rank and sum-rank metrics. A notable example is is the work of Guruswami and Xing~\cite{guruswami2013list}, which presents an explicit list-decodable Reed--Solomon subcode based on the \textit{subspace design}. This approach subsequently adapted to the rank-metric setting to obtain efficiently list-decodable Gabidulin subcodes~\cite{guruswami2016explicit}.

    These results highlight both the recent progress and remaining challenges in list decoding of LRS codes under the sum-rank metric, thereby motivating our work on explicit sum-rank metric code constructions that admit efficient list decoding beyond the unique-decoding radius.
    
\subsection{Overview of Main Results}
     In this work, we give the first explicit construction of $\F_h$-linear subcodes of LRS and FLRS codes that are efficiently list-decodable beyond the unique-decoding radius under the sum-rank metric, with provably bounded output list size.
\begin{theorem}
    For every $\eps\in(0,1)$ and integer $s>0$, there exists an explicit $\Fh$-linear subcode of the linearized Reed--Solomon code $\mathrm{LRS}[\balpha,\bbeta;\boldsymbol{n},k]\subseteq\F_{{h}^{t}}^{n}$ with rate at least $(1-2\eps)k/n$, where the evaluation points span $\F_{h^n}$ and the messages are taken from $\F_{h^t}[x;\sigma]$ with $\sigma(x)=x^h$.
    
    Moreover, the code is list-decodable in polynomial time from up to $\frac{s}{s+1}(n-k)$ sum-rank errors. The output list is contained in an $\Fh$-subspace of dimension at most $O(s^2/\eps^2)$, and hence has size at most $h^{O(s^2/\eps^2)}$.
\end{theorem}
This proof is inspired by the argument of Guruswami, Wang, and Xing in~\cite{guruswami2016explicit}. Broadly, we construct a family of LRS subspace codes in which each message vector $(f_0,\ldots,f_{k-1})$ is chosen from a product of subspaces $H_1 \times \cdots \times H_k$ forming a subspace design. Unlike the Hamming-metric setting where one can often take the alphabet size comparable to the block length, in the rank/sum-rank setting the relevant ambient field is typically an extension field whose size grows exponentially with the block length. To address this issue, we employ a suitable generalization of subspace designs, as suggested in~\cite{guruswami2016explicit} for the case of Gabidulin codes. Moreover, that work provides an explicit construction of such generalized subspace designs, which ensures that our resulting LRS subcode construction remains explicit.

We develop a linear-algebraic list-decoding algorithm inspired by the linear-algebraic frameworks of~\cite{guruswami2013list} for folded Reed--Solomon and derivative codes, and of~\cite{guruswami2016explicit} for Gabidulin codes. The algorithm proceeds in two main steps.

First, we perform an interpolation step to construct a skew polynomial $Q$ satisfying
$$Q\left(\beta_{i,j},y_{i,j},y_{i,j}^q,y_{i,j}^{q^2},\cdots,y_{i,j}^{q^{s-1}}\right)_{a_i}=0$$
for all positions corresponding to the received word $\boldsymbol{y}=(\boldsymbol{y}_1\mid\boldsymbol{y}_2\mid\cdots\mid\boldsymbol{y}_{\ell})\in\Fqmn$, where $\boldsymbol{y}_i=(y_{i,1},y_{i,2},\cdots,y_{i,n_{i}})\in\Fqm^{n_i}$. 

Then, we apply a linear-algebraic algorithm to compute an affine subspace that contains all message vectors satisfying the resulting algebraic decoding condition. We show that this solution space has a periodic subspace structure and can be computed in polynomial time. The final list of candidate messages is obtained by intersecting this periodic subspace with the message space of our code. Since the message space is chosen from a subspace design and periodic subspaces have limited intersection with such designs, the resulting list size is provably small. Since both the periodic subspace and the code are $\F_h$-linear, their intersection can be computed in time polynomial in $t$ and $\log h$. Moreover, the intersection is an $\F_h$-subspace of dimension at most $O(ms/\varepsilon)$, and hence the resulting list size is at most $h^{O(ms/\varepsilon)}$.

The following is our construction of construction of a subcode of the folded linearized Reed--Solomon code.

\begin{theorem}
For any $\eps\in(0,1)$ and integer $1\le s\le \lambda$, there exists an explicit subcode of the $\lambda$-folded linearized Reed--Solomon code $\mathrm{FLRS}[\boldsymbol{a},\gamma;\lambda,n,k]\subseteq\F_{q^m}^{\lambda\times n}$ of block length $n=N/\lambda$ and rate at least $(1-\varepsilon)k/N$
such that the code is list-decodable in polynomial time from up to
  \[
    e\ \le\ \frac{s}{s+1}\left(\frac{n(\lambda-s+1)-k+1}{\lambda-s+1}\right)
  \]
  of sum-rank errors, with output list size at most $(d/\varepsilon)^d$, where $d=\lceil k/m\rceil(s-1)=O(s-1)$.
\end{theorem}

As in the LRS setting, the central idea is to restrict the message space so that its intersection with the algebraically characterized solution space produced by the decoder is provably small. Given a received word, we first perform an interpolation step that produces algebraic constraints satisfied by all valid FLRS message polynomials. These constraints characterize an affine solution space over $\Fqm$ of dimension at most $O(s-1)$, which can be computed in polynomial time. To control the list size, we restrict the message space to an explicit subspace-evasive set constructed in~\cite{guruswami2013linear}. Finally, by intersecting the low-dimensional affine solution space with the message space, we obtain the output list size is at most $(d/\varepsilon)^d$, where $d=\lceil k/m\rceil(s-1)=O(s-1)$.

Furthermore, one may instead employ the Monte Carlo construction of subspace-evasive sets due to Guruswami~\cite{guruswami2011linear}, which yields a substantially smaller list size.

\begin{theorem}
    For any $\eps > 0$, there exists a Monte Carlo construction of a subcode
    of the $\lambda$-folded linearized Reed--Solomon code
    $\mathcal{C} = \mathrm{FLRS}[\balpha,\gamma;\lambda,n,k]$ of block length $n = N/\lambda$ such
    that, with high probability, $\mathcal{C}$ has rate at least
    $(1-\eps)k/N$ and is list-decodable in polynomial time from a fraction
    \[
    \frac{s}{s+1}\left(\frac{n(\lambda-s+1)-k+1}{\lambda-s+1}\right)
    \]
    of sum-rank errors, for all integers $1 \le s \le m$, with output list size at most
    $O(s/\eps)$.
\end{theorem}


\subsection{Organization of the Paper}	
In Section~\ref{sec:pre}, we present the preliminaries needed for this paper, including some definitions and known results about sum-rank metric codes, the construction method of Linearized Reed--Solomon codes, as well as the definitions of subspace designs and periodic subspaces. In Section~\ref{sec:list_decoding}, we describe our linear-algebraic list-decoder and analyze its output in order to show how to use subspace designs to construct list-decodable sum-rank metric codes. The construction of list-decodable subspace codes appears in Section~\ref{sec:construction}. Section~\ref{sec:FLRS} extends our list-decoding framework beyond standard LRS 
codes to \emph{folded linearized Reed--Solomon codes}. We establish the corresponding interpolation conditions, and analyze the structure of the solution space. We then demonstrate that the same subspace-design-based argument yields explicit FLRS subcodes that are efficiently list-decodable beyond the unique decoding radius.

\section{Preliminaries}\label{sec:pre}
\subsection{Sum-Rank Metric Codes}
Sum-rank metric measures \textit{tuples of matrices} $\F_q^{m_1\times {n}_1}\times\cdots\times\F_q^{m_{\ell}\times {n}_{\ell}}$, where $m_1,\cdots,m_{\ell}$, ${n}_1,\cdots,n_{\ell}$ are positive integers satisfying ${n}_i\leq m_i$ for $1\leq i\leq\ell$, and $q$ is a prime power. Let $\X=(X_1,\ldots,X_{\ell}),\Y=(Y_1,\ldots,Y_{\ell})\in\F_q^{m_1\times {n}_1}\times\cdots\times\F_q^{m_{\ell}\times {n}_{\ell}}$, their \emph{weight} and \emph{distance} are defined as
$$
\mathrm{wt}_{\mathsf{SR}}(X)\coloneqq\sum_{i=1}^\ell \rank(X_i),\quad d_{\mathsf{SR}}(X,Y)\coloneqq\sum_{i=1}^\ell \rank(X_i-Y_i).
$$
We focus on the uniform-parameter setting in which $m_i=m$ 
for all $1\leq i\leq\ell$ and let $n=\sum_{i=1}^{\ell}n_i$. Throughout this paper, fix an $\F_q$-basis of $\F_{q^m}$, which induces an $\F_q$-linear isomorphism $\F_{q^m}^{n_i} \cong \F_q^{m \times n_i}$. Under this identification, $\F_q^{m \times n_1} \times \cdots \times \F_q^{m \times n_\ell}$ can be viewed as $\F_{q^m}^{n}$. Therefore, the weight and distance defined above induce a metric on $\Fqmn$, called the \emph{sum-rank metric}. When $\ell=n$, it reduces to the Hamming metric, and when $\ell=1$, to the rank metric. 

For a sum-rank metric code $\cod\subseteq\Fqmn$, its \emph{rate} and \emph{minimum sum-rank distance}(hereafter simply referred to as the distance) are defined as 
\[R\coloneqq\frac{\log_q|\cod|}{n},\quad\dsr(\cod):=\min\limits_{\Xs\neq\Ys}\{\dsr(\Xs,\Ys):\Xs,\Ys\in\cod\}.\]

To formalize the notion of list-decodability for sum-rank metric codes, we first define sum-rank metric spheres and balls. For a center $\Xs\in\Fqmn$ and a radius $0\leq r\leq n$, the sum-rank sphere and ball are defined as
\[\Ssr(\Xs,r)=\{\Ys\in\Fqmn\mid\dsr(\Ys,\Xs)= r\},\]\[\Bsr(\Xs,r)=\{\Ys\in\Fqmn\mid\dsr(\Ys,\Xs)\leq r\}.\]
Using these definitions, we can now define list-decodable sum-rank metric codes.
\begin{definition}[List-decodable sum-rank metric code]
	 For $\rho\in(0,1)$ and an integer $L\geq1$, a sum-rank metric code $\cod\subseteq\Fqmn$ is said to be $(\rho,L)$-list-decodable if for any $\Xs\in\Fqmn$, 
	$$
	|\Bsr(\Xs,\rho n)\cap\cod|\leq L.
	$$
\end{definition}
\subsection{Skew polynomials}
Analogous to Reed--Solomon codes, linearized Reed--Solomon (LRS) codes are defined via the evaluation of skew polynomials. In contrast to the conventional polynomial ring, the ring of skew polynomials is defined with respect to an automorphism $\sigma:\Fqm\rightarrow\Fqm$. When considering the field extension $\Fqm/\Fq$, we take $\sigma$ to be the Frobenius automorphism $\sigma(x)=x^q$.

The skew polynomial ring $\Fqm[x;\sigma]$ consists of polynomials $f=\sum_{i}f_ix^i$ where addition is defined as
\[f+g=\sum_{i}(f_i+g_i)x^i,\]
and multiplication is defined by
\[f\cdot g=\sum_{i}\sum_{j}f_i\sigma^i(g_j)x^{i+j}.\]

Multiplication in $\Fqm[x;\sigma]$ is non-commutative whenever $\sigma$ is not an identity transformation. The degree of a skew polynomial $f=\sum_{i}f_ix^i$ is defined as
\[\deg f\coloneqq
\begin{cases}
    \max\{i:f_i\neq0\}, & \text{if }f\neq0,\\
    -\infty, & \text{otherwise.}
\end{cases}\]
For an integer $k$, we denote by $\Fqm[x;\sigma]_{<k}$ the set of skew polynomials of degree less than $k$. 

For any $a,b\in\Fqm$, define the operator
\[\mathcal{D}_a(b)\coloneqq\sigma(b)a.\]
For an integer $i\geq0$, define recursively
\[\mathcal{D}_{a}^{i+1}(b)\coloneqq\mathcal{D}_a\left(\mathcal{D}_a^i(b)\right)=\sigma^{i+1}(b)N_{i+1}(a)\]
where $\mathcal{D}^0_a(b) = b$ and $$N_{i}(a):=\prod_{j=0}^{i-1}\sigma^j(\alpha)$$is called the \textit{generalized power function}~\cite{lam1988vandermonde}. 

For an integer $i<0$, define
\[\mathcal{D}_a^{-i}(b):=\sigma^{-i}(b)\left(\sigma^{-i}\left(N_{i}(a)\right)\right)^{-1}.\]
Observe that for any integers $i, j$, the operators satisfy
\[\mathcal{D}_a^j\left(\mathcal{D}_a^i(b)\right)=\mathcal{D}_a^{i+j}(b).\]
Linearized Reed--Solomon codes are evaluation codes of skew polynomials, where the evaluation map is given by the \emph{generalized operator evaluation}.
\begin{definition}[Definition~20, \cite{martinez2018skew}]\label{def:generalized_operator_evaluation}
    The \textit{generalized operator evaluation} of a skew polynomial $f\in\Fqm[x,\sigma]$ at $b\in\Fqm$ with respect to $a\in\Fqm$ is defined as
\[f(b)_a\coloneqq\sum_{i}f_i\mathcal{D}_a^i(b)=\sum_{i}f_i\sigma^i(b)N_i(a).\]
\end{definition}

Since $\sigma$ is linear, the generalized operator evaluation defines an $\Fq$-linear map. Specifically, for any $f \in\Fqm[x,\sigma]$, $\beta,\gamma\in\Fq$ and $a, b, c\in\Fqm$, we have
\[f(\beta b+ \gamma c)_a = \beta f(b)_a+\gamma f(c)_a.\]

Two elements $a,b\in\Fqm$ are called $\sigma$-conjugates if there exists an element $c\in\Fqm^*$,  such that
\[b=\sigma(c)ac^{-1}.\]
Elements that are not $\sigma$-conjugates are called $\sigma$-distinct. The relation of $\sigma$-conjugacy defines an equivalence relation on $\Fqm$ and hence induces a partition of $\Fqm$ into conjugacy classes. The conjugacy classes of $a$ is defined as
\[C_{\sigma}(a)=\{\sigma(c)ac^{-1}\mid c\in\Fqm^*\}.\]

The fundamental properties of Reed--Solomon codes rely on Lagrange interpolation over the conventional polynomial ring. However, such interpolation does not directly extend to skew polynomial rings due to their non-commutative structure. Consequently, the evaluation points for linearized Reed--Solomon codes must satisfy additional constraints, which we describe next.
\begin{lemma}[\!\cite{bartz2021decoding}]\label{lem:lagrange_interpolation}
    Let $b_{i,1},b_{i,2},\cdots,b_{i,n_i}$ be $\Fq$-linearly independent elements of $\Fqm$ for $i=1,2,\cdots,\ell$. Let $c_{i,1},c_{i,2},\cdots,c_{i,n_i}\in\Fqm$ and let $a_1,a_2,\cdots,a_{\ell}$ be representatives of pairwise distinct $\sigma$-conjugacy classes of $\Fqm$. 
    Then there exists a unique skew polynomial $f\in\Fqm[x;\sigma]$ such that
    \[f(b_{i,j})_{a_i}=c_{i,j},\forall i=1,\cdots,\ell,\forall j=1,2,\cdots,n_{i},\]
    and $\deg(f)\leq \sum_{i=1}^\ell n_i$.
\end{lemma}
The following lemma characterizes the maximum number of interpolation constraints that a nonzero skew polynomial of a given degree can satisfy. In particular, it provides a sharp bound on the number of roots (counted with respect to generalized operator evaluation) and complements Lemma~\ref{lem:lagrange_interpolation}.
\begin{lemma}[Proposition~1, \cite{bartz2021decoding}]\label{lem:root}
    Let $b_{i,1},b_{i,2},\cdots,b_{i,n_i}\in\Fqm$ and let $a_1,a_2,\cdots,a_{\ell}$ be representatives be of $\sigma$-conjugacy classes of $\Fqm$ for $i=1,2,\cdots,\ell$. Then for any nonzero skew polynomial $f\in\Fqm[x;\sigma]$ satisfying
    \[f(b_{i,j})_{a_i}=c_{i,j},\forall i=1,\cdots,\ell,\forall j=1,2,\cdots,n_{i},\]
    we have that $\deg(f)\leq \sum_{i=1}^\ell n_i$ where equality holds if and only if the $b_{i,1},b_{i,2},\cdots,b_{i,n_i}$ are $\Fq$-linearly independent for each $i=1,2,\cdots,\ell$. 
\end{lemma}

\subsection{Subspace designs}
Subspace designs were first introduced in~\cite{guruswami2013list} as a combinatorial tool for constructing list-decodable subcodes of Reed--Solomon codes in the Hamming metric.
\begin{definition}[Definition~6, \cite{guruswami2013list}]
    Let $m$ be a positive integer and let $q$ be a prime power. A collection $S$ of $\Fq$-subspaces $H_1,H_2,\cdots,H_M\subseteq\Fq^m$ is called an $(s,A)$ $\Fq$-\emph{subspace design} if for every $\Fq$-linear subspace $W\subseteq\Fq^m$ of dimension $s$,
    \[\sum_{i=1}^M\dim(H_i\cap W)\leq A.\]
\end{definition}
In the case of Reed--Solomon codes, the field size $q$ can be chosen to be comparable to the block length of the code. As a result, the constructions in~\cite{guruswami2016explicit2} yield explicit subcodes that are efficiently list-decodable.

In contrast, for rank-metric and subspace codes--and similarly for sum-rank metric codes--the effective field size is exponential in the code length. Motivated by this limitation, Guruswami and Wang~\cite{guruswami2016explicit} introduced a generalization of subspace designs that is better suited to these settings, which we define next.
\begin{definition}\label{def:subspace_design}
    Let $m,n$ be positive integers and let $q$ be a prime power. A collection $S$ of $\Fq$-subspaces $H_1,H_2,\cdots,H_M\subseteq\Fq^{mn}$ is called an $(s,A,n)$ $\Fq$-\emph{subspace design} if for every $\Fqn$-linear space $W\subseteq\Fqn^m$ of dimension $s$,
    \[\sum_{i=1}^M\dim_{\Fq}(H_i\cap W)\leq A.\]
    Here we fix an $\F_q$-basis of $\F_{q^n}$, which induces an $\F_q$-linear isomorphism
$\F_{q^n}^m \cong \F_q^{mn}$. Under this identification, we view $W\subseteq \F_{q^n}^m$ as an $\F_q$-subspace of $\F_q^{mn}$.
\end{definition}
Selecting several subspaces from the subspace design results in a small intersection between their Cartesian product and the periodic subspaces (to be introduced next), see~\ref{coll:intersection}. This is crucial for enabling our algorithm to output a compact list.

For a vector $\Xs=(x_1,x_2,\cdots,x_n)$, we denote by $\text{proj}_{[a,b]}(\Xs)$ the vector $(x_a,x_{a+1},\cdots,x_b)$ for $a,b\in [n]$.
\begin{definition}[Definition~1, \cite{guruswami2013list}]\label{def:periodic_subspaces}
     For positive integers $s,m,k$, an affine subspace $H\subseteq\Fq^{mk}$ is said to be $(s,m,k)$-\textbf{periodic} if there exists a subspace $W\subseteq\Fq^m$ with $\dim(W)\leq s$ such that for every $j=1,2,\cdots,k$, and every prefix $\boldsymbol{a}\in\Fq^{(j-1)m}$, the projected affine subspace of $\Fq^m$ defined by
     $$
     \{\text{proj}_{[(j-1)m+1,j\cdot m]}(\Xs)\mid\Xs\in H\text{ and }\text{proj}_{[1,(j-1)\cdot m]}(\Xs)=\boldsymbol{a}\}
     $$
     is contained in an affine subspace, which can be denoted by $W+\boldsymbol{v_a}$ where $\boldsymbol{v_a}\in\Fq^m$ only depends on $\boldsymbol{a}$.
\end{definition}
To achieve diagrammatic clarity in subsequent proofs, we give a canonical representation of periodic subspace.
\begin{definition}\label{def:canonical}
    The canonical representation of an $(s,m,k)$-periodic subspace $H$ consists of a matrix $B\in\Fq^{m\times m}$ such that $\ker(B)$ has dimension at most $s$, and vectors $a_i\in\Fq^m$ and matrices $A_{i,j}\in\Fq^{m\times m}$ for $1\leq i\leq k$ and $1\leq j<i$, such that $\Xs\in H$ if and only if for every $i=\os k$ the following holds:
    \[a_i+\left(\sum_{j=1}^{i-1}A_{i,j}\cdot\text{proj}_{[(j-1)m+1,j\cdot m]}(\Xs)\right)+B\cdot\text{proj}_{[(i-1)m+1,i\cdot m]}(\Xs)=0.\]
\end{definition}
\begin{corollary}[Proposition~11, \cite{guruswami2016explicit}]\label{coll:intersection}
    Let $H$ be a $(s,m,k)$-periodic affine subspace of $\Fqn^{mk}$, and $H_1,H_2,\cdots,H_k\subseteq\Fq^{mn}$ be distinct subspaces from a $(s,A,n)$ $\Fq$-subspace design. Then $H\cap (H_1\times\cdots\times H_k)$ is an affine space over $\Fq$ of dimension at most $A$.
\end{corollary}

\begin{proof}
    Let $W$ be the subspace associated to $H$ as in Definition~\ref{def:periodic_subspaces}. We will show by induction that \[\left|\text{proj}_{[1,im]}(H\cap(H_1\times\cdots\times H_i))\right|\leq q^{\sum_{j=1}^i\text{dim}_{\F_{q}}(W\cap H_j)}.\]
    
    For $i=1$, $\text{proj}_{[1,m]}(H)\cap(H_1)=(W+v_0)\cap H_1$ is an affine subspace whose underlying subspace lies in $W\cap H_1$, hence it has size at most $q^{\dim(W\cap H_1)}$.
    
    Continuing, fix an $\boldsymbol{a}$. As $H$ is periodic and $H_{i+1}$ is linear, the possible extensions of $\boldsymbol{a}$ in $\text{proj}_{[im+1,(i+1)m]}(H) \cap H_{i+1}$ are given by a coset of $W\cap H_{i+1}$. Thus, there are at most $q\cdot \dim(W\cap H_{i+1})$ such extensions. Therefore, the number of possible prefixes $\boldsymbol{a}$ is at most $q^{\sum_{j=1}^i\dim(W\cap H_j)}$.
    
    In particular, $H\cap (H_1\times\cdots\times H_k)$ has dimension at most $\sum_{i=1}^k\dim(W\cap H_i)\leq A$ over $\Fq$, by the subspace design property.
\end{proof}
    Subspace-evasive sets were introduced by Guruswami~\cite{guruswami2013linear} for constructing explicit list-decodable subcodes of folded Reed–Solomon and derivative codes. We adopt this approach to construct our FLRS subcodes.
    \begin{definition}[Subspace-Evasive Sets,~\cite{guruswami2013linear}]\label{def:evasive-sets}
        A subset $S\subseteq\Fq^k$ is $(s,L)$-subspace-evasive if for every $\Fq$-linear affine subspace $H\subseteq\Fq^k$ of dimension $s$, $|H\cap S|\leq L$.
    \end{definition}

\section{Explicit Construction of List-Decodable Linearized Reed--Solomon Codes}\label{sec:list_decoding}
Now we turn to list decoding linearized Reed–Solomon codes via a linear-algebraic approach.
\begin{definition}[Linearized Reed--Solomon Codes, \cite{martinez2018skew}]
    Let $m,n,\ell$ be positive integers, and let partition $n=n_1+n_2+\cdots+n_{\ell}$, denoted by $\boldsymbol{n}=(n_1,\cdots,n_{\ell})$, which induces a sum rank metric on $\Fqmn$ with $\ell<q$. Let $\sigma$ be an automorphism of $\Fqm$ whose fixed field is $\Fq$. Fix an integer $1\leq k\leq n$ and an \emph{evaluation pair} 
    $(\balpha,\bbeta)\in\Fqm ^\ell\times\Fqm^n$, where $\balpha=(a_1,\cdots,a_{\ell})$ and $\bbeta=(\bbeta_1\mid\cdots\mid\bbeta_{\ell})$, where $\bbeta_{i}=(\beta_{i,1},\beta_{i,2},\cdots,\beta_{i,n_i})\in\Fqm^{n_i}$. Assume that the elements $a_1,\cdots,a_{\ell}$ lie in distinct $\sigma$ conjugacy classes, and that the elements of each $\bbeta_i$ are linearly independent over $\Fq$. The corresponding linearized Reed–Solomon code is defined as
    \[LRS[\balpha,\bbeta;\boldsymbol{n},k]:=\{f(\bbeta)_{\balpha}=\left(f(\bbeta_1)_{a_1}\mid f(\bbeta_2)_{a_2}\mid\cdots\mid f(\bbeta_{\ell})_{a_{\ell}}\right)\mid f\in\Fqm[x;\sigma]_{<k}\},\]
    where, for $1\leq i\leq \ell$, \[f(\bbeta_i)_{a_i}=\left(f(\beta_{i,1})_{a_i},f(\beta_{i,2})_{a_i},\cdots,f(\beta_{i,n_i})_{a_i}\right)\in\Fqm^{n_i}\]for $1\leq i\leq \ell$.
\end{definition}
It is shown in~\cite{martinez2018skew,caruso1908residues} that, for any evaluation $(\balpha,\bbeta)$, the evaluation map $f\mapsto f(\bbeta)_{\balpha}$ is bijective. Moreover, for any nonzero polynomial $f\in\Fqm[x;\sigma]_{<k}$, the sum-rank weight satisfies $\wsr(f(\bbeta)_{\balpha})\geq n-\deg f$. Consequently, the minimum sum-rank distance of $\mathrm{LRS}[\balpha,\bbeta;\boldsymbol{n},k]$ is $d=n-k+1$, achieving the Singleton bound for the sum-rank metric~\cite{martinez2018skew}.
These codes coincide with Gabidulin codes when $\ell=n$, and with generalized Reed–Solomon codes when $\ell=1$.

In this section, we show that by restricting the evaluation points to a suitable subfield, the output list of the decoder forms a periodic subspace. Our decoding algorithm is inspired by the linear-algebraic list decoding techniques developed in~\cite{guruswami2013linear,guruswami2016explicit}. 

Let $n,t$ be integers and $h$ is a prime power. Consider $\mathrm{LRS}[\balpha,\bbeta;\boldsymbol{n},k]\subseteq\F_{{h}^{t}}^{n}$ encoding the skew polynomial over $\F_{h^t}[x;\sigma]$ where $\sigma(x)=x^h$. Assume that $n\mid t$ and write $t=nm$ so that $\F_{{h}^{t}}$ contains a subfield $\F_{h^n}$. We choose the evaluation pair $(\balpha,\bbeta)$ from this subfield $\F_{h^n}$. Let $q=h^n$, then under this assumption, $\F_{{h}^{n}}=\F_{q}$, $\F_{{h}^{t}}=\F_{{q}^m}$. 

Now suppose that codeword $\boldsymbol{c}=f(\bbeta)_{\balpha}\in\Fqmn$ is transmitted,  and that the received word is $\boldsymbol{y}=(\boldsymbol{y}_1\mid\boldsymbol{y}_2\mid\cdots\mid\boldsymbol{y}_{\ell})\in\Fqmn$ with at most $e$ sum-rank errors, where $\boldsymbol{y}_i=(y_{i,1},y_{i,2},\cdots,y_{i,n_{i}})\in\Fqm^{n_i}$. 

\subsection{Interpolation Step}\label{sec:Interpolation_step}
Choose an integer parameter $1\leq s\leq m$, and define
\[D=\left\lfloor\frac{n-k+1}{s+1}\right\rfloor.\]
We consider multivariate skew polynomials of the form
    \begin{align}\label{eq:Q}                            
        Q(x,y_1,y_2,\cdots,y_s)=Q_0(x)+Q_1(x)\cdot y_1+Q_2(x)\cdot y_2+\cdots+Q_s(x)\cdot y_s,
    \end{align}
where $Q_j\in\Fqm[x;\sigma]$ for all $j=0,1,\cdots,s$. The collection of all such polynomials is denoted by $\Fqm[x,y_1,y_2,\cdots,y_s;\sigma]$. 

For a skew polynomial $f(x)=f_0+f_1 x +\cdots+ f_{k-1} x^{k-1}\in\Fqm[x;\sigma]_{<k}$, define its $q$-Frobenius twist as $$f^{(q)}(x)=f_0 ^q+f_1^qx +\cdots+f_{k-1}^qx^{k-1}.$$
\begin{claim}
    For $a,\beta\in\Fq$, we have $f^{(q)}(\beta)_{a}=\left(f(\beta)_{a}\right)^q$, and more generally, $f^{(q^i)}(\beta)_{a}=\left(f(\beta)_{a}\right)^{q^i}$.
\end{claim}
\begin{proof}
    \begin{align*}
        f^{(q)}(\beta)_{a}&=\sum_{i}f_i^q\mathcal{D}_a^i(\beta)
        =\sum_{i}f_i^q\left(\mathcal{D}_a^i(\beta)\right)^q\\
        &=\left(\sum_{i}f_i\mathcal{D}_a^i(\beta)\right)^q
        =\left(f(\beta)_{a}\right)^q.
    \end{align*}
The second equality follows from the fact that $\mathcal{D}_a^i(\beta)\in\Fq$, and that the map $x\mapsto x^q$ is an automorphism of $\Fqm$ fixing $\Fq$. The general statement follows immediately by induction.
\end{proof}
This observation explains why the evaluation pair must be chosen from a subfield of the coefficient field: it ensures compatibility between Frobenius powers and generalized operator evaluation.

Let $\mathcal{P}$ be the space of skew polynomials $Q\in\Fqm[x,y_1,y_2,\cdots,y_s;\sigma]$ of the form~\eqref{eq:Q}, such that $\deg(Q_0)\leq D+k-1$ and $\deg(Q_i)\leq D$ for $i=\os s$. Our goal is to analyze expressions of the form \[Q(x,f(x),f^{(q)}(x),\cdots,f^{(q^{s-1})}(x))\] for suitable $Q\in\mathcal{P}$. The interpolation relies on the following product rule.
\begin{lemma}[Product rule, \cite{martinez2019private}]\label{lem:product_rule}
    Let $f,g\in\Fqm[x;\sigma]$ be skew polynomials and let $a,\beta\in\Fqm$. Then the generalized operator evaluation satisfies
    \[(f\cdot g)(\beta)_a=f(g(\beta)_a)_a.\]
\end{lemma}

By Lemma~\ref{lem:product_rule}, we obtain
\begin{equation}
\begin{aligned}
Q\!\left(\beta, f(\beta), f^{(q)}(\beta), \ldots, f^{(q^{s-1})}(\beta)\right)_a
&= Q_0(\beta)_a \\
&\quad + Q_1\!\bigl(f(\beta)_a\bigr)_a
+ \cdots
+ Q_s\!\bigl(f^{(s-1)}(\beta)_a\bigr)_a .
\end{aligned}
\end{equation}

Motivated by the fact that each received symbol $y_{i,j}$ corresponds to an operator evaluation $f(\beta_{i,j})_{a_i}$, we define the \emph{$n$ generalized operator evaluation} of a polynomial $Q\in\Fqm[x,y_1,y_2,\cdots,y_s;\sigma]$ by
\[Q\left(\beta,\gamma_1,\gamma_2\cdots\gamma_s\right)_{a}=Q_0(\beta)_a+Q_1(\gamma_1)_a+\cdots+Q_s(\gamma_s)_a,\]
for $a, \beta,\gamma_1,\cdots,\gamma_s\in\Fqm$, 
    
The following lemma guarantees the existence of a nonzero interpolation polynomial. 
\begin{lemma}\label{lem:interpolation}
    Let $\boldsymbol{y}=(\boldsymbol{y}_1\mid\boldsymbol{y}_2\mid\cdots\mid\boldsymbol{y}_{\ell})\in\Fqmn$, where $\boldsymbol{y}_i=(y_{i,1},y_{i,2},\cdots,y_{i,n_{i}})\in\Fqm^{n_i}$, and let $(\balpha,\bbeta)$ be an evaluation pair. Then there exists a nonzero polynomial $Q\in\mathcal{P}$ such that
    \begin{equation}\label{eq:condition_Q}
        Q\left(\beta_{i,j},y_{i,j},y_{i,j}^q,y_{i,j}^{q^2},\cdots,y_{i,j}^{q^{s-1}}\right)_{a_i}=0
    \end{equation}
    for $i=\os \ell$ and $j=\os n_i$.
    Moreover, such a polynomial can be found using $O(n^2)$ operations over $\Fqm$.
\end{lemma}
\begin{proof}
    For each pair $(i,j)$,~\eqref{eq:condition_Q} imposes an $\Fqm$-linear constraint on the coefficients of polynomials $Q_0,\cdots,Q_s$. The dimension of the vector space $\mathcal{P}$ is
    \[(D+k)+s(D+1)=(D+1)(s+1)+k-1>n,\]
    whereas the total number of interpolation constraints equals $n$. Hence, the resulting homogeneous linear system admits a nontrivial solution. 
    
    Such a polynomial $Q$ can be computed via Newton interpolation~\cite[Proposition 2.6]{martinez2022codes} with complexity $O(n^2)$ over $\Fqm$.
\end{proof}

\subsection{Finding the candidates}\label{sec:finding_the_candidates}
The following lemma shows that the skew polynomial $Q$ given by Lemma~\ref{lem:interpolation} yields an algebraic condition, which the output candidate of the list decoder must satisfy.
\begin{lemma}\label{lem:key-equation}
Let $f\in\Fqm[x;\sigma]$ be a skew polynomial of degree at most $k-1$. Suppose the codeword $\boldsymbol{c}=(\boldsymbol{c}_1\mid\cdots\mid\boldsymbol{c}_\ell)\in\Fqmn$ with $\boldsymbol{c}_i=(f(\beta_{i,1})_{a_i},\ldots,f(\beta_{i,n_i})_{a_i})$, is received as $\boldsymbol{y}=(\boldsymbol{y}_1\mid\cdots\mid\boldsymbol{y}_\ell)\in\Fqmn$ with $\boldsymbol{y}_i=(y_{i,1},y_{i,2},\cdots,y_{i,n_{i}})$ with at most $e$ sum-rank errors. If $e\le \frac{s}{s+1}(n-k)$, then
\begin{align}\label{eq:f}
Q_0(x)+\sum_{j=1}^s Q_j(x)\cdot f^{(q^{j-1})}(x)=0.
\end{align}
\end{lemma}

\begin{proof}
    Define
    \[
    \varPhi(x)=Q(x,f(x),f^{(q)}(x),\ldots,f^{(q^{s-1})}(x)).
    \]
    
    As $h^t=q^m$, let $A_i,B_i\in\Fh^{n_i\times 2t}$ be defined as
    \[
    A_i=\begin{pmatrix}
    \beta_{i,1} & f(\beta_{i,1})_{a_i}\\
    \vdots & \vdots\\
    \beta_{i,n_i} & f(\beta_{i,n_i})_{a_i}
    \end{pmatrix},
    \quad
    B_i=\begin{pmatrix}
    \beta_{i,1} & y_{i,1}\\
    \vdots & \vdots\\
    \beta_{i,n_i} & y_{i,n_i}
    \end{pmatrix},
    \]
    for each $i=1,\ldots,\ell$, where we identify elements of $\F_{h^t}$ with vectors in $\F_h^t$ via a fixed basis.
    
    We denote by $\left \langle X\right \rangle$ the $\Fh$-spaces spanned by the rows of $X$. We will use the following lemma:
    \begin{lemma}[Lemma~12, \cite{guruswami2016explicit}]\label{lem:rank}
        Let $X,Y\in \Fh^{n\times t}$ with $\rank(X-Y)\leq e$, then $\dim_{\Fh}(\left \langle X\right \rangle\cap\left \langle Y\right \rangle )\geq\dim_{\Fh}(\left \langle X\right \rangle)-e$
    \end{lemma}
    Applying Lemma~\ref{lem:rank} to the matrices $A_i$ and $B_i$, we obtain that $\rank(A_i-B_i)=\rank(\boldsymbol{c}_i-\boldsymbol{y}_i)$. Suppose $\rank(\boldsymbol{c}_i-\boldsymbol{y}_i)=e_i$, then $\sum_{i=1}^{\ell}e_i=e$.
    
    By Lemma~\ref{lem:rank}, $\dim_{\Fh}(\left\langle A_i\right\rangle\cap\left\langle B_i\right\rangle)\geq\dim_{\Fh}(\left\langle A_i\right\rangle)-e_i=n_i-e_i$. This implies there exists an $\Fh$-subspace $U_i\subseteq \spa\{\beta_{i,1},\cdots,\beta_{i,n_i}\}$ of dimension at least $n_i-e_i$ such that, for $\beta_i=\sum_{j=1}^{n_i}c_j\beta_{i,j}\in U_i$ with $c_j\in \Fh$, one has $\sum_{j=1}^{n_i}c_jf(\beta_{i,j})_{a_i}=\sum_{j=1}^{n_i}c_jy_{i,j}$. Hence, 
    \begin{align*}
        0&=\sum_{j=1}^{n_i}c_jQ(\beta_{i,j},y_{i,j},y_{i,j}^q,\cdots,y_{i,j}^{q^{s-1}})_{a_i}\\
        &=Q\left(\beta_i,f(\beta_{i,j})_{a_i},\left(f(\beta)_{a_i}\right)^q,\left(f(\beta_{i,j})_{a_i}\right)^{q^2},\cdots,\left(f(\beta_{i,j})_{a_i}\right)^{q^{s-1}}\right)_{a_i}\\
        &=Q(\beta_i, f(\beta_i), f^{(q)}(\beta_i), f^{(q^2)}(\beta_i),\cdots, f^{(q^{s-1})}(\beta_i))_{a_i}\\
        &=\varPhi(\beta_i)_{a_i}
    \end{align*}
    This implies that for $i=1,\cdots,\ell$, there are at least $n_i-e_i$ $\Fh$-linear independent $\beta_i$ such that $\varPhi(\beta_i)_{a_i}=0$. Thus,
    \begin{align*}
        \sum^{\ell}_{i=1}(n_i-e_i)=n-e>D+k-1\geq \deg(\varPhi).
    \end{align*}
    According to Lemma~\ref{lem:root}, $\varPhi$ can only be a zero polynomial.
\end{proof}

So the candidate messages emerge from the solution space of $f$ that satisfies~(\ref{eq:f}). The following implies that the solutions form a periodic subspace.
\begin{lemma}\label{lem:solutions}
    The set of solutions $f(x)=f_0+f_1 x +\cdots+ f_{k-1} x^{k-1}\in\Fqm[x;\sigma]$ to the equation
    \begin{align}\label{eq:solution}
        Q_0(x)+Q_1(x)\cdot f(x)+Q_2(x)\cdot f^{(q)}(x)+\cdots+Q_s(x)\cdot f^{(q^{s-1})}(x)=0
    \end{align}
    (where at least one of $\{Q_0,Q_1,\cdots,Q_s\}$ is nonzero) is an $(s-1,m,k)$-periodic subspace.
\end{lemma}

\begin{proof}
Define the $\F_q$-linear map
\[
T:\F_{q^m}[x;\sigma]_{<k}\to \F_{q^m}[x;\sigma],\qquad
T(f):=Q_0+\sum_{j=1}^s Q_j\cdot f^{(q^{j-1})}.
\]
Then $f$ satisfy \eqref{eq:solution} is equivalent to $T(f)=0$. Since $T$ is $\F_q$-linear in $f$, the solution set $\{f:T(f)=0\}$ is an $\F_q$-affine subspace. We now proceed to analyze its structure.


    First, we can assume that at least one $Q^{\star}_i(x)$ has nonzero constant term, for some $i^{*}\in\{0,1,\cdots,s\}$, since we can left-divide by $x^d$ for some $d$. Moreover, if $Q_1(x),\cdots,Q_s(x)$ is all divisible by $x$, then so is $Q_0(x)$, thus we can take $i^*>0$.

    For $i=0,1,2,\cdots,s$, denote
    \[Q_{i}(x)=q_{i,0}+q_{i,1}x+q_{i,2}x^2+\cdots.\]
    For $j=0,1,2,\cdots,k-1$, define the linearized polynomial
    \begin{align}\label{eq:linearized polynomial}
        R_{j}(x)=q_{1,j}x+q_{2,j}x^q+q_{3,j}x^{q^2}+\cdots+q_{s,j}x^{q^{s-1}}.
    \end{align}
    Since $R_0$ is a nonzero $q$-linearized polynomial of $q$-degree at most $s-1$, its root set in $\F_{q^m}$ is an $\F_q$-vector subspace of dimension at most $s-1$, which we denote by $W$.
    
    Fix an $i\in\{0,1,\cdots,k-1\}$. Expanding the~\eqref{eq:solution} and equating the coefficient of $x^i$ to be $0$, we have
    \begin{align}\label{eq:xi}
        q_{0,i}+R_{i}\left(\sigma^i(f_0)\right)+R_{i-1}\left(\sigma^{i-1}(f_1)\right)+\cdots+R_1\left(\sigma(f_{i-1})\right)+R_0(f_i)=0.
    \end{align}
    This implies $f_i\in W+b_i$ for some $b_i\in\Fqm$ that is determined by $f_0,f_1,\cdots,f_{i-1}$. So for each choice of $f_0,f_1,\cdots,f_{i-1}$, $f_i$ must belong to a fixed coset of the subspace $W$ of dimension at most $s-1$. Thus by Definition~\ref{def:periodic_subspaces}, we conclude that the solution space is an $(s-1,m,k)$-periodic subspace. Also, it is clear from Definition~\ref{def:canonical} that a canonical representation of the periodic subspace can be computed in $\poly(k,m,\log q)$ time.
\end{proof}

\subsection{Code Construction}\label{sec:construction}

We now present an explicit construction of list-decodable subcodes of linearized Reed--Solomon (LRS) codes by combining the list-decoding framework developed in Sections~\ref{sec:Interpolation_step} and~\ref{sec:finding_the_candidates} with explicit subspace designs. Throughout, we work under the parameter setting of Section~\ref{sec:list_decoding}: the code alphabet is $\F_{h^t}$ with $t=nm$, the automorphism is $\sigma(x)=x^h$, and the evaluation pair $(\alpha,\beta)$ is chosen from the subfield $\F_{h^n}\subseteq \F_{h^t}$.

For every received word $\boldsymbol{y}$, we can compute, in polynomial time, an affine solution space of message polynomials satisfying the key algebraic constraint imposed by the decoder. Crucially, this solution space admits a strong structural description: it forms an $(s-1,m,k)$-periodic subspace in the coefficient domain (see Lemma~\ref{lem:solutions}). Consequently, in order to bound the final list size, it suffices to ensure that the message space intersects any such periodic subspace in only a small number of points.

To this end, we recall the notion of subspace designs (see Definition~\ref{def:subspace_design}). A collection of $\F_q$-subspaces $\{H_1,\dots,H_M\}$ of $\F_q^{mn}$ is an $(s,A,n)$ $\F_q$-subspace design if, for every $\F_{q^n}$-linear subspace $W\subseteq \F_{q^n}^m$ of $\F_{q^n}$-dimension $s$, we have
\[
\sum_{i=1}^M \dim_{\F_q}(H_i\cap W)\ \le\ A.
\]
Intuitively, this guarantees that any low-dimensional ``structured'' subspaces $W$ cannot have large cumulative intersections with many of the subspaces $H_i$'s.

The relevance of subspace designs to our decoding framework is formalized in Corollary~\ref{coll:intersection}: Specifically, 
if $\mathcal{P}\subseteq (\F_{q^n}^m)^k$ is an $(s-1,m,k)$-periodic affine subspace and $H_1,\dots,H_k$ are chosen from an $(s,A,n)$ subspace design, then
\[
\dim_{\F_q}\big(\mathcal{P}\cap(H_1\times\cdots\times H_k)\big)\ \le\ A.
\]
Therefore, restricting the message space to $H_1\times\cdots\times H_k$ ensures that the final candidate list produced by the decoder is small.

We use the explicit construction of subspace designs due to Guruswami and Kopparty, stated below in a form tailored to our parameter regime.

\begin{theorem}[Theorem~8, \cite{guruswami2016explicit}]\label{thm:ex_subspace_design}
Fix $\varepsilon\in(0,1)$ and integers $m,n,s$ with $s\le \varepsilon m/4$.
Let $q$ be a prime power such that $q^n>m$.
Then there exists an \emph{explicit} collection of $\F_q$-subspaces of $\F_q^{mn}$ of size $q^{\Omega(\varepsilon mn/s)}$ forming an $(s,\, 2(m-1)s/\varepsilon,\, n)$ $\F_q$-subspace design. Moreover, bases for these subspaces can be computed in time polynomial in the relevant parameters.
\end{theorem}

Let $\{H_1,\dots,H_k\}$ be $k$ distinct subspaces selected from the subspace design guaranteed by Theorem~\ref{thm:ex_subspace_design}.
We restrict the message polynomial
\[
f(x)=\sum_{i=0}^{k-1} f_i x^i
\]
by enforcing
\[
f_i\in H_{i+1}\qquad\text{for all }i=0,1,\dots,k-1.
\]
This restriction yields an explicit $\F_h$-linear subcode of the original LRS code. Since bases for the subspaces $H_i$ can be computed efficiently (see Theorem~\ref{thm:ex_subspace_design}), the encoding algorithm is efficient as well.

\begin{theorem}[Explicit list-decodable LRS subcodes]\label{thm:explicit-lrs-subcode}
For any $\eps\in(0,1)$ and integer $1\leq s\leq m$, there exists an explicit $\Fh$-linear subcode of the linearized Reed--Solomon code $\mathrm{FRS}[\balpha,\bbeta;\boldsymbol{n},k]\subseteq\F_{{h}^{t}}^{n}$ of rate at least $(1-2\eps)k/n$ such that:
\begin{enumerate}
  \item \textbf{Decoding radius:} The code is list-decodable in polynomial time from up to
  \[
  e\ \le\ \frac{s}{s+1}(n-k)
  \]
  sum-rank errors.
  \item \textbf{List size:} The output list is contained in an $\F_h$-subspace of dimension at most $O(s^2/\varepsilon^2)$, and hence has list size at most $h^{O(s^2/\varepsilon^2)}$.
\end{enumerate}
\end{theorem}

\begin{proof}
Fix $\varepsilon\in(0,1)$ and an integer $s$ with $1\le s\le m$, and define the subcode $\mathcal{C}$ as described above. Fix a received word $\boldsymbol{y}$.

\smallskip
\noindent\textbf{Step 1: Structure of the decoder output.}
We run the interpolation and candidate-finding procedures from Sections~\ref{sec:Interpolation_step} and ~\ref{sec:finding_the_candidates}. By Lemma~\ref{lem:interpolation}, there exists a nonzero interpolation polynomial satisfying the constraints~\eqref{eq:condition_Q}. By Lemma~\ref{lem:solutions}, the set of solutions to the corresponding key algebraic equation forms an $(s-1,m,k)$-periodic affine subspace in the coefficient domain, which can be computed in polynomial time. We denote this subspace by $\mathcal{P}(\boldsymbol{y})$.

\smallskip
\noindent\textbf{Step 2: Correctness within the decoding radius.}
Suppose the transmitted codeword differs from $\boldsymbol{y}$ in at most
$e \le \frac{s}{s+1}(n-k)$ sum-rank errors. Then by Lemma~\ref{lem:key-equation}, the transmitted message polynomial $f$ satisfies the key equation derived from the interpolation step. Consequently, its coefficient vector lies in the affine subspace $\mathcal{P}(\boldsymbol{y})$.

\smallskip
\noindent\textbf{Step 3: Intersection with the restricted message space.}
By construction of the subcode, valid messages correspond precisely to coefficient vectors in
\[
\mathcal{M}\ :=\ H_1\times\cdots\times H_k.
\]
Therefore, the final output list produced by the decoder is exactly $\mathcal{P}(\boldsymbol{y})\cap \mathcal{M}$.

By Corollary~\ref{coll:intersection}, since $\mathcal{P}(\boldsymbol{y})$ is $(s-1,m,k)$-periodic and $\{H_1,\dots,H_k\}$ are chosen from an $(s,A,n)$ subspace design with
\[
A=\frac{2(m-1)s}{\varepsilon},
\]
we have
\[
\dim_{\F_h}\big(\mathcal{P}(\boldsymbol{y})\cap \mathcal{M}\big)\ \le\ A.
\]
Thus, the output list is contained in an $\F_h$-subspace of dimension at most $2(m-1)s/\varepsilon$, implying a list size of at most $h^{O(ms/\varepsilon)}$.

Finally, by choosing $m=O(s/\varepsilon)$, we obtain an output dimension bounded by $O(s^2/\varepsilon^2)$, and hence a list size of at most $h^{O(s^2/\varepsilon^2)}$.
\end{proof}

\begin{remark}
    By choosing $s=O(1/\eps)$, we obtain a code of rate $R=(1-2\eps)k/n$ that is list-decodable from up to $(1-R-\eps)n$ sum-rank errors, with list size bounded by $h^{\poly(1/\eps)}$.
\end{remark}

\section{Explicit Construction of List-Decodable Folded Linearized Reed--Solomon Codes}\label{sec:FLRS}
    In this section, we establish explicit list-decodable folded linearized Reed--Solomon codes. Throughout this section, we work in the parameter regime $m=\Theta(n)$, under which the proposed list-decoding algorithm runs in polynomial time in the block length. We begin by recalling the definition of folded linearized Reed--Solomon (FLRS) codes.
    \begin{definition}[Folded Linearized Reed--Solomon Codes]
        Let $\gamma$ be a primitive element of $\F_{q^m}$, and let
        $
        \balpha=(a_1,\ldots,a_{\ell})\in\F_{q^m}^{\ell}
        $
        be such that the elements $a_1,\ldots,a_{\ell}$ belong to distinct conjugacy classes.
        Let $N:=\lambda n$ denote the (unfolded) number of evaluation points. A $\lambda$-folded linearized Reed--Solomon code of length $n:= \frac{N}{\lambda}$ and dimension $k\leq n$ is defined as follows:
        \[
        \mathrm{FLRS}[\balpha,\gamma;\lambda,n,k]
        :=\left\{\boldsymbol{C}(f)
        =\bigl(\boldsymbol{C}_1(f)\mid\cdots\mid\boldsymbol{C}_{\ell}(f)\bigr)
        \;\middle|\;
        f\in\F_{q^m}[x;\sigma]_{<k}
        \right\},
        \]
        where for each $i=1,\ldots,\ell$,
        \[
        \boldsymbol{C}_i(f):=
        \begin{pmatrix}
        f(1)_{a_i}            & f(\gamma^{\lambda})_{a_i}          & \cdots & f(\gamma^{(\eta-1)\lambda})_{a_i} \\
        f(\gamma)_{a_i}       & f(\gamma^{\lambda+1})_{a_i}        & \cdots & f(\gamma^{(\eta-1)\lambda+1})_{a_i} \\
        \vdots                & \vdots                        & \ddots & \vdots \\
        f(\gamma^{\lambda-1})_{a_i} & f(\gamma^{2\lambda-1})_{a_i}        & \cdots & f(\gamma^{\eta \lambda-1})_{a_i}
        \end{pmatrix}
        \in \F_{q^m}^{\,\lambda\times \eta},
        \]
        and $\eta := n/\ell$.
    \end{definition}
    We now consider the code $\mathrm{FLRS}[\balpha,\gamma;\lambda,n,k]$, which encodes a message polynomial
    $f\in\F_{q^m}[x;\sigma]_{<k}$ with $\sigma(x)=x^q$.
    Suppose that the transmitted codeword $\boldsymbol{C}(f)\in\F_{q^m}^{\lambda\times n}$ is received as
    \[
    \boldsymbol{y}
    =\bigl(\boldsymbol{y}_1\mid\boldsymbol{y}_2\mid\cdots\mid\boldsymbol{y}_{\ell}\bigr)
    \in\F_{q^m}^{\lambda\times n},
    \]
    with at most $e$ sum-rank errors, where for each $i=1,\ldots,\ell$,
    \[
    \boldsymbol{y}_i=
    \begin{pmatrix}
    y_{i,0}     & y_{i,\lambda}     & \cdots & y_{i,(\eta-1)\lambda} \\
    y_{i,1}     & y_{i,\lambda+1}   & \cdots & y_{i,(\eta-1)\lambda+1} \\
    \vdots      & \vdots      & \ddots & \vdots \\
    y_{i,\lambda-1}   & y_{i,2\lambda-1}  & \cdots & y_{i,\eta \lambda-1}
    \end{pmatrix}
    \in\F_{q^m}^{\lambda\times \eta}.
    \]

    \subsection{Interpolation}\label{sec:flrs-interpolation}
    We employ a Welch-Berlekamp style interpolation procedure, following the approach used for folded Reed--Solomon codes~\cite{guruswami2011linear}. We perform $(s+1)$-variate skew polynomial interpolation with respect to an interpolation parameter $s\in\N^*$ satisfying $s\leq \lambda$. 
    Let
    \[
    D=\left\lfloor\frac{n(\lambda-s+1)-k+1}{s+1}\right\rfloor,
    \]
    and let $\mathcal{L}\subseteq\F_{q^m}[x,y_1,\ldots,y_s;\sigma]$ denote the set of all skew polynomials of the form~\eqref{eq:Q}
    such that
    \[
    \deg(Q_0)\leq D+k-1
    \quad\text{and}\quad
    \deg(Q_i)\leq D,\qquad 1\leq i\leq s.
    \]
    
    We perform interpolation as follows.
    \begin{lemma}\label{lem:interpolation2}
    Consider the code $\mathrm{FLRS}[\balpha,\gamma;\lambda,n,k]$ and a received word
    \[
    \boldsymbol{y}
    =(\boldsymbol{y}_1\mid\boldsymbol{y}_2\mid\cdots\mid\boldsymbol{y}_{\ell})
    \in\F_{q^m}^{\lambda\times n}.
    \]
    There exists a nonzero polynomial $Q\in\mathcal{L}$ satisfying the interpolation conditions
    \begin{align}\label{eq:interpolation2}
    Q\bigl(\gamma^{i\lambda+j},
          y_{k,i\lambda+j},
          y_{k,i\lambda+j+1},
          \ldots,
          y_{k,i\lambda+j+s-1}
    \bigr)_{a_k}=0,
    \end{align}
    for all
    \[
    i=0,1,\ldots,\eta-1,\quad
    j=0,1,\ldots,\lambda-s,\quad
    k=1,\ldots,\ell.
    \]
    \end{lemma}
    \begin{proof}
    Each condition in~\eqref{eq:interpolation2} yields a linear equation over $\F_{q^m}$ in the coefficients of $Q$.
    The total number of interpolation constraints is
    \[
    \eta (\lambda-s+1)\ell = n(\lambda-s+1).
    \]
    On the other hand, the total number of unknown coefficients is
    \[
    (D+1)s + D + k
    = (D+1)(s+1)+k-1,
    \]
    which is strictly larger than $n(\lambda-s+1)$ by the choice of $D$.
    Hence, there exists a nonzero polynomial $Q\in\mathcal{L}$ satisfying all interpolation conditions.
    \end{proof}
    \begin{remark}
         The interpolation problem~\eqref{eq:interpolation2} can be solved via skew K\"otter interpolation from~\cite{liu2014kotter,bartz2014efficient} in $O(sN^2)$ operations in $\Fqm$. Alternatively, fast interpolation algorithms~\cite{bartz2024fast,bartz2022fast} achieve a complexity of
         $\widetilde{O}(s^{\omega}\mathcal{M}(n))$ operations in $\Fqm$, where $\mathcal{M}(n)=O(n^{1.635})$ is the cost of skew-polynomial multiplication and $\omega< 2.37286$ is the matrix multiplication exponent~\cite{bronstein1993international}.
    \end{remark}
    
    Given a polynomial $Q$ satisfying the interpolation conditions~\eqref{eq:interpolation2}, the following lemma shows that any codeword within the decoding radius must satisfy a certain algebraic condition.
    \begin{lemma}[Theorem~3, \cite{hormann2024interpolation}]\label{lem:interpolation3}
    Let $Q$ be a nonzero polynomial obtained from Lemma~\ref{lem:interpolation2}.
    If the number of sum-rank errors $e$ satisfies
    \[
    e\leq
    \frac{s}{s+1}\cdot
    \frac{n(\lambda-s+1)-k+1}{\lambda-s+1},
    \]
    then the transmitted message polynomial $f\in\F_{q^m}[x;\sigma]_{<k}$ satisfies
    \begin{align}\label{eq:f2}
    Q_0(x)
    +Q_1(x)\cdot f(x)
    +Q_2(x)\cdot f(x)\cdot\gamma
    +\cdots
    +Q_s(x)\cdot f(x)\cdot\gamma^{s-1}
    =0.
    \end{align}
    \end{lemma}
\subsection{Finding Candidate Solutions}\label{sec:flrs-finding-solutions}
We now determine the solution space of~\eqref{eq:f2}. The following lemma shows that the set of solutions forms an affine subspace of bounded dimension.

\begin{lemma}\label{lem:solutions2}
The set of solutions
\[
f(x)=f_0+f_1x+\cdots+f_{k-1}x^{k-1}\in\F_{q^m}[x;\sigma]
\]
satisfying
\begin{align}\label{eq:solution2}
Q_0(x)+Q_1(x)\cdot f(x)
+Q_2(x)\cdot f(x)\cdot \gamma
+\cdots
+Q_s(x)\cdot f(x)\cdot \gamma^{s-1}=0,
\end{align}
where at least one of $\{Q_0,Q_1,\ldots,Q_s\}$ is nonzero,
forms an affine $\Fqm$-subspace of dimension at most $O(s-1)$. Moreover, one can compute this using $O(N^2)$ field operations over $\Fqm$. 
\end{lemma}

\begin{proof}
    Similar to Lemma~\ref{lem:solutions}, since~\eqref{eq:solution2} is linear in the coefficients of $f$ over $\Fqm$, the set of solutions forms an affine $\Fqm$-subspace. We now bound its dimension. 
    
    Without loss of generality, we may assume that there exists $i^\ast\in\{1,\ldots,s\}$ such that skew polynomial $Q_{i^\ast}(x)$ has a nonzero constant term. 

    For $i=0,1,2,\cdots,s$, denote
    \[Q_{i}(x)=q_{i,0}+q_{i,1}x+q_{i,2}x^2+\cdots.\]
    For $j=0,1,\ldots,k-1$, define the polynomial
    \begin{align}\label{eq:R_polynomial2}
        R_j(x)=q_{1,j}+q_{2,j}x+q_{3,j}x^2+\cdots+q_{s,j}x^{s-1}.
    \end{align}
    
    Since $q_{i^*,0}\neq0$, it follows that $R_0\neq0$. By condition, for $r\in\{0,1,\cdots,k-1\}$, the coefficient of $x^r$ of~\eqref{eq:solution2} is 0. The constant term equals to 
    \[q_{0,0}+q_{1,0}f_0+q_{2,0}f_0\gamma+\cdots+q_{s,0}f_0\gamma^{s-1}=q_{0,0}+R_0(\gamma)f_0.\]
    If $R_0(\gamma)\neq0$, then $f_0$ uniquely determined as $-q_{0,0}/R_0(\gamma)$. If $R_0(\gamma)=0$, then $q_{0,0}=0$ or else there will be no solutions to~\eqref{eq:solution2} and in that case $f_0$ may take an arbitrary value in $\Fqm$.

    Next, equating the coefficient of $x^r$ in~\eqref{eq:solution2} for $0\leq r<k$, we obtain
    
    \begin{equation}\label{eq:linear_form}
        \sum_{j=0}^{r} \sum_{i=0}^{s}q_{i,j} \, \sigma^{j}(f_{\,r-j}) \, \sigma^{r}(\gamma^{\,i})=R_0(\sigma^r(\gamma)) f_r+\sum_{j=1}^{r} R_j(\sigma^{r-j}(\gamma)) \sigma^j(f_{r-j})+q_{0,r}=0.
    \end{equation}
    If $R_0(\sigma^r(\gamma))\neq 0$, then $f_r$ is uniquely determined by $f_0,\ldots,f_{r-1}$. Otherwise, $f_r$ may take an arbitrary value in $\F_{q^m}$.

    Thus, the dimension of the solution space is at most the number of indices \[r\in\{0,1,\cdots,k-1\} \quad s.t.\quad R_0(\sigma^r(\gamma))=0.\]Since $\sigma^r(\gamma)=\gamma^{q^r}$ and $\gamma^{q^m}=\gamma$, the sequence $\gamma^{q^r}$ repeat with period $m$. Moreover, since $\deg(R_0)=s-1<m$, among the $m$ distinct elements $\gamma,\gamma^q,\ldots,\gamma^{q^{m-1}}$, the equation $R_0(\gamma^{q^r})=0$ can hold for at most $s-1$ values of $r$. Consequently, among the indices $r=0,\ldots,k-1$, there are at most \[\lceil k/m\rceil(s-1)\] values for which $R_0(\sigma^r(\gamma))=0$. Hence, the solution space is an affine $\F_{q^m}$-subspace of dimension at most $d:=\lceil k/m\rceil(s-1)$. Since $m=\Theta(n)$, we have $\lceil k/m\rceil(s-1)=O(s-1)$.

    Finally, the claimed $O(N^2)$ complexity follows from the fact that the linear system defined by~\eqref{eq:linear_form} has a simple lower-triangular structure.
\end{proof}
     Combining Lemmas~\ref{lem:interpolation2} and~\ref{lem:solutions}, we can conclude the following.
    \begin{theorem}\label{thm:LFRS_code}
        Suppose the $\mathrm{FLRS}[\balpha,\gamma;\lambda,n,k]$ of block length $n=N/\lambda$, rate $R=k/N$, encodes polynomials from $f\in \F_{q^m}[x;\sigma]$ with $\sigma(x)=x^q$. Let integer $1\leq s\leq m$. Given a received word $\boldsymbol{y}\in\Fqm^{\lambda\times n}$, then one can find, in time $O(N^2m^2)$ over $\Fq$, a basis of an affine $\F_{q^m}$-subspace of dimension $O(s-1)$ that contains all message polynomials $f\in\F_{q^m}[x;\sigma]_{<k}$ whose codeword differs from $\boldsymbol{y}$ in at most errors of
        \[\frac{s}{s+1}\cdot\frac{n(\lambda-s+1)-k+1}{\lambda-s+1}\]
     \end{theorem}

\subsection{Code Construction}\label{sec:flrs-construction}

We now present an explicit construction of list-decodable subcodes of folded linearized Reed–Solomon (FLRS) codes by combining the list-decoding framework developed in Sections~\ref{sec:flrs-interpolation} and~\ref{sec:flrs-finding-solutions} with explicit subspace-evasive sets. Throughout this subsection, we work in the parameter regime of Section~4 and assume $m=\Theta(n)$.

Fix a received word $\boldsymbol{y}\in\F_{q^m}^{\lambda\times n}$. The interpolation step described in Section~\ref{sec:flrs-interpolation}, together with the candidate-finding step from Section~\ref{sec:flrs-finding-solutions}, produces in polynomial time a nonzero interpolation polynomial $Q$ (see Lemma~\ref{lem:interpolation2}). Moreover, any transmitted message polynomial $f\in\F_{q^m}[x;\sigma]_{<k}$ whose codeword lies within the decoding radius must satisfy the key algebraic equation~\eqref{eq:f2}.

Crucially for controlling the list size, Theorem~\ref{thm:LFRS_code} shows that from $\boldsymbol{y}$ one can compute an affine $\F_{q^m}$-subspace
\[
\mathcal{P}(\boldsymbol{y})\ \subseteq\ \F_{q^m}^k
\]
of dimension $O(s-1)$ that contains the coefficient vectors $(f_0,\ldots,f_{k-1})$ of all such valid message polynomials $f$. Hence, in the FLRS setting, the list-decoding problem reduces to intersecting a bounded-dimensional affine subspace $\mathcal{P}(\boldsymbol{y})$ with a suitably restricted message set.

To ensure that this intersection is small for every low-dimensional affine subspace, we employ \emph{subspace-evasive sets}. Recall that a set $S\subseteq \F_q^k$ is called \emph{$(s,L)$-subspace-evasive} if for every affine subspace $H\subseteq\F_q^k$ of dimension at most $s$, one has $|S\cap H|\le L$ (see Definition~\ref{def:evasive-sets}). Therefore, once the decoder outputs an affine candidate space $\mathcal{P}(y)$ of dimension at most $d=\lceil k/m\rceil(s-1)$, restricting messages to an $(d,L)$-subspace-evasive set guarantees that the final list size is bounded by $L$.

We use the explicit construction of subspace-evasive sets due to Dvir and Lovett, stated below in a form tailored to our parameter regime.

\begin{theorem}[Explicit subspace-evasive sets, \cite{dvir2012subspace}]\label{thm:explicit_evasive_set}
For any $\varepsilon>0$, integer $s\ge 1$, and prime power $q$, there exists an explicit set $S\subseteq \F_q^k$ with
\[
|S|\ >\ q^{(1-\varepsilon)k}
\]
that is $(s,\,(s/\varepsilon)^s)$-subspace-evasive. Moreover, $S$ can be constructed in time polynomial in $\log q$ and $k$.
\end{theorem}

We identify a message polynomial $f(x)=\sum_{i=0}^{k-1} f_i x^i$ with its coefficient vector $(f_0,\ldots,f_{k-1})\in\F_{q^m}^k$. Choose an explicit $(d,(d/\varepsilon)^d)$-subspace-evasive set
\begin{equation}\label{S-choose}
    S\ \subseteq\ \F_{q^m}^k
\end{equation}
Since $S$ is explicit and efficiently constructible, the resulting encoding procedure is explicit as well.

\begin{theorem}[Explicit list-decodable FLRS subcodes]\label{thm:flrs-subcode-explicit}
For any $\eps\in(0,1)$ and integer $1\le s\le \lambda$, there exists an explicit subcode of the $\lambda$-folded linearized Reed--Solomon code $\mathrm{FLRS}[\boldsymbol{a},\gamma;\lambda,n,k]\subseteq\F_{q^m}^{\lambda\times n}$ of block length $n=N/\lambda$ and rate at least $(1-\varepsilon)k/N$
such that:
\begin{enumerate}
  \item \textbf{Decoding radius:} The code is list-decodable in polynomial time from up to
  \[
    e\ \le\ \frac{s}{s+1}\left(\frac{n(\lambda-s+1)-k+1}{\lambda-s+1}\right)
  \]
  of sum-rank errors.
  \item \textbf{List size:} The output list size is at most $(d/\varepsilon)^d$, where $d=\lceil k/m\rceil(s-1)=O(s-1)$.
\end{enumerate}
\end{theorem}

\begin{proof}

Fix $\varepsilon\in(0,1)$ and an integer $s$ with $1\le s\le \lambda$, and let $d:=\lceil k/m\rceil(s-1)$. Let $S\subseteq\mathbb{F}_{q^m}^k$ be an explicit $(d,(d/\varepsilon)^d)$-subspace-evasive set guaranteed by \eqref{S-choose}, and define the subcode $\mathcal{C}$ by restricting the message polynomials $f(x)=\sum_{i=0}^{k-1} f_i x^i\in\F_{q^m}[x;\sigma]_{<k}$ to those whose coefficient vector $(f_0,\ldots,f_{k-1})$ lies in $S$. Fix a received word $\boldsymbol{y}\in\F_{q^m}^{\lambda\times n}$.

\smallskip
\noindent\textbf{Step 1: Structure of the decoder output.}
We run the list-decoding algorithm of Sections~\ref{sec:flrs-interpolation} and~\ref{sec:flrs-finding-solutions}. By Lemma~\ref{lem:interpolation2}, there exists a nonzero interpolation polynomial satisfying the constraints~\eqref{eq:interpolation2}. Moreover, by Lemma~\ref{lem:solutions2}, the set of solutions to the corresponding key algebraic equation forms an affine subspace in the coefficient domain. By Theorem~\ref{thm:LFRS_code}, this affine subspace has dimension at most $d$ and can be computed in polynomial time. We denote this subspace by $\mathcal{P}(\boldsymbol{y})$.

\smallskip
\noindent\textbf{Step 2: Correctness within the decoding radius.}
Suppose the transmitted codeword differs from $\boldsymbol{y}$ in at most
$e\ \le\ \frac{s}{s+1}\left(\frac{n(\lambda-s+1)-k+1}{\lambda-s+1}\right)$ sum-rank errors. 
According to Lemma~\ref{lem:interpolation3}, the affine solution space $\mathcal{P}(\boldsymbol{y})$ contains the coefficient vectors of all message polynomials consistent with $\boldsymbol{y}$ within the decoding radius. 

\smallskip
\noindent\textbf{Step 3: Intersection with the restricted message space.}
As the construction of the subcode $\mathcal{C}$, the set of candidate valid messages consistent with $\boldsymbol{y}$ is exactly $\mathcal{P}(\boldsymbol{y})\cap S$. Since $\dim(\mathcal{P}(\boldsymbol{y}))\leq d$ and $S$ is $(d,(d/\varepsilon)^d)$-subspace-evasive, we have
\[
|\mathcal{P}(\boldsymbol{y})\cap S|\ \le\ (d/\varepsilon)^d.
\]
Moreover, this intersection can be computed efficiently in time polynomial in its size; thus the overall decoding time remains polynomial in the block length (for $m=\Theta(n)$) and in the list bound.
\end{proof}

\begin{remark}
        By choosing $s = O(1/\eps)$ and $\lambda= O(1/\eps^2)$, the resulting code can be list decoded in polynomial time from a fraction $1-R-\eps$ of errors, while producing an output list of size at most $(1/\eps^2)^{O(1/\eps)}$.
    \end{remark}
    
    Moreover, one may instead employ the Monte Carlo construction of subspace-evasive sets due to Guruswami~\cite{guruswami2011linear}, which yields a substantially smaller list size.

    \begin{theorem}[Lemma~9, \cite{guruswami2011linear}]\label{thm:pseudorandom_evasive_set}
    Let $q$ be a prime power and let $k$ be a positive integer. For any $\eps>0$ and any integer $s$ satisfying $1 \leq s \leq \eps k/2$, there exists a Monte Carlo construction of a set $S \subseteq \Fq^k$ with cardinality $|S| > q^{(1-\eps)k}/2$ that is $(s, O(s/\eps))$-subspace evasive with probability at least $1 - q^{-\Omega(k)}$.
    \end{theorem}
    \begin{theorem}
        For any $\eps > 0$, there exists a Monte Carlo construction of a subcode
        of the $\lambda$-folded linearized Reed--Solomon code
        $\mathcal{C} = FLRS[\balpha,\gamma;\lambda,n,k]$ of block length $n = N/\lambda$ such
        that, with high probability, $\mathcal{C}$ has rate at least
        $(1-\eps)k/2N$ and is list-decodable in polynomial time from a fraction
        \[
        \frac{s}{s+1}\left(\frac{n(\lambda-s+1)-k+1}{\lambda-s+1}\right)
        \]
        of sum-rank errors, for all integers $1 \le s \le \lambda$, with output list size at most
        $O(s/\eps)$.
    \end{theorem}
    In particular, by choosing $s = O(1/\eps)$ and $\lambda = O(1/\eps^2)$, the output list size can be bounded by $O(1/\eps^2)$.

\bibliographystyle{plain}
\bibliography{main}
\appendix

\end{document}